\documentclass[conference]{IEEEtran}
\IEEEoverridecommandlockouts
\usepackage{optidef}
\usepackage{amsmath,amsfonts}
\usepackage{algorithm}
\usepackage[inline]{enumitem}
\usepackage{array}
\usepackage[caption=false,font=normalsize,labelfont=sf,textfont=sf]{subfig}
\usepackage{textcomp}
\usepackage{stfloats}
\usepackage{url}
\usepackage{verbatim}
\usepackage{amsmath,amssymb,amsfonts}
\usepackage{mathtools, physics, bbm}
\usepackage{amsthm}
\usepackage{cite}
\newtheorem{theorem}{Theorem}
\usepackage{graphicx}
\usepackage{xcolor}

\newtheorem{proposition}{Proposition}
\newtheorem{definition}{Definition}

\newtheorem{corollary}{Corollary}
\newtheorem{lemma}{Lemma}

\newtheorem{example}{Example}

\usepackage{graphicx}

\theoremstyle{remark}
\newtheorem{remark}{Remark}

\theoremstyle{plain}
\usepackage{mathrsfs}
\usepackage{xspace}

\usepackage{cite}
\usepackage{amsmath,amssymb,amsfonts}
\usepackage{algorithmic}
\usepackage{graphicx}
\usepackage{textcomp}
\usepackage{xcolor}
\newcommand{\cC}{\mathcal{C}}

\definecolor{bleudefrance}{rgb}{0.19, 0.55, 0.91}

\def\BibTeX{{\rm B\kern-.05em{\sc i\kern-.025em b}\kern-.08em
    T\kern-.1667em\lower.7ex\hbox{E}\kern-.125emX}}

\newcommand{\supp}{\mathrm{supp}}
\newcommand{\wt}{\mathrm{wt}}

\newcommand{\rtxc}{\ensuremath{(r,t,x)}-cLRC\xspace}
\newcommand{\rtxcs}{\ensuremath{(r,t,x)}-cLRCs\xspace}
\newcommand{\rtxq}{\ensuremath{(r,t,x)}-qLRC\xspace}
\newcommand{\rtxqs}{\ensuremath{(r,t,x)}-qLRCs\xspace}
\begin{document}

\title{CSS Quantum LRCs with Intersecting Recovery Sets: Constructions and Bounds

\thanks{H.\ B. is also affiliated with the Munich Quantum Valley (MQV), the Munich Center for Quantum Science and Technology (MCQST)}
}

\author{%
  \IEEEauthorblockN{Evagoras Stylianou$^1$, Vinayak Ramkumar$^2$, Holger Boche$^{1}$ and Rawad Bitar$^2$}
  \IEEEauthorblockA{$^1$Chair of Theoretical Information Technology, Technical University of Munich, \\  
   $^2$Institute for Communications Engineering, Technical University of Munich, \\
  Email: \{evagoras.stylianou, vinayak.ramkumar, boche, rawad.bitar\}@tum.de \vspace{-0.68cm}
}
}

\maketitle

\begin{abstract}
In this work, we study \((r,t,x)\) quantum locally recoverable codes (qLRCs) with locality \(r\), \(t\) recovery sets per qudit, and intersection parameter \(x\). We first show that, assuming the underlying classical codes have dual minimum distance at least two, a CSS code is an \((r,t,x)\)-qLRC if and only if the underlying classical codes are \((r,t,x)\) classical LRCs (cLRCs) with common recovery sets. We then use subset-inclusion matrices to construct families of binary dual-containing \((r,t,x)\)-cLRCs, which yield binary \((r,t,x)\)-qLRCs via the CSS construction. For CSS \((r,t,x)\)-qLRCs, we derive upper bounds on the dimension and rate, minimum-distance bounds in the pure case, and a Singleton-like dimension bound in the exact case. Finally, we show that these families attain high rates and nontrivial minimum distances.\vspace{-0.2cm}
\end{abstract}

\section{Introduction}
Quantum locally recoverable codes (qLRCs) provide a
framework for quantum error correction under locality constraints: the erasure of any single qudit can be corrected by accessing at most $r$ other qudits. We refer to these qudits together with the erased qudit as a \emph{recovery set}. Since the pioneering work of Golowich and Guruswami~\cite{golowich2023quantum}, several bounds, constructions, and generalizations for qLRCs have been developed~\cite{sharma2025quantum,luo2025bounds,galindo2026quantum,zhou2025optimal,xie2025two,li2025improved,cao2025optimal,bu2025quantum}.
For classical LRCs (cLRCs), an important strengthening of locality is availability, where each symbol has multiple recovery sets that intersect only in the symbol itself~\cite{rawat2016locality}. Throughout, we refer to such recovery sets as \emph{disjoint}. In the quantum setting, however, already two disjoint recovery sets for the same qudit force it to be in a fixed product state with the rest of the code, and hence to carry no quantum information about the encoded state~\cite{golowich2023quantum}. Thus, a direct quantum analog of classical availability is impossible.

Motivated by this limitation, Bu--Gu--Li~\cite{bu2025quantum} introduced $(r,t,x)$-qLRCs, the quantum analogue of $(r,t,x)$-cLRCs~\cite{kruglik2017one}, by allowing multiple \emph{intersecting} recovery sets. Here, each qudit has $t$ recovery sets with pairwise intersection size at most $x+1$. By a slight abuse of terminology, we refer to $t$ as the availability. They also defined \emph{exact} $(r,t,x)$-qLRCs, for which all recovery sets have size $r+1$, all pairwise intersections have size $x+1$, and any three intersect only at the erased qudit. They derived a dimension bound for general $(r,t,x)$-qLRCs, a Singleton-like dimension bound for the exact case, and constructed an exact family via the CSS framework.


In the classical setting, the literature on $(r,t,x)$-cLRCs is less developed
than in the disjoint case $(x=0)$. In~\cite{kruglik2017one,kruglik2018distance,
kruglik2019new}, it was observed that allowing intersections between recovery
sets can improve the code rate relative to the disjoint setting. Using the
recovery graph framework~\cite{tamo2016bounds}, the authors
of~\cite{kruglik2017one} derived a rate bound for \rtxcs and gave constructions
by repeating columns of the parity-check matrix of the Wang--Zhang--Liu (WZL)
construction of $(r,t,0)$-cLRCs~\cite{wang2015achieving}. Alphabet-dependent
distance bounds were later obtained in~\cite{kruglik2018distance}. Further
constructions, based on parity-check matrices from subspace inclusion relations
and on extensions of the Tamo--Barg construction, were given
in~\cite{teng2021constructions,wang2016two,rajput2020rs,rajput2022subclass}.


In this work, we study $(r,t,x)$-qLRCs obtained via the CSS construction. Our
main contributions are as follows. 
\begin{enumerate}[label=\emph{(\roman*)}]
\itemsep0.1em
\item We show that, under the assumption that the underlying classical codes have dual minimum distance at least two, a CSS code is an $(r,t,x)$-qLRC iff the underlying classical codes are $(r,t,x)$-cLRCs with common recovery~sets.
\item We show that a family of binary codes defined by subset-inclusion matrices gives $(r,t,x)$-cLRCs with explicit locality, availability, and intersection parameter. We then characterize exactly when these codes are dual-containing, thereby obtaining explicit families of binary CSS $(r,t,x)$-qLRCs.
\item For CSS \rtxqs, we derive dimension and rate bounds, minimum-distance bounds in the pure case, and a Singleton-like dimension bound in the exact case.
\end{enumerate}
Finally, we compare an exact subfamily of the subset-inclusion construction with the construction of Bu--Gu--Li~\cite{bu2025quantum}, which is, to the best of our knowledge, the only previously known explicit \rtxq construction. \vspace{-0.15cm}

\section{Preliminaries} \label{sec:pre}
Let $q$ be a prime power, and use the convention that $\binom{a}{b}=0$ whenever $b<0$ or $b>a$. For a positive integer $n$, set $[n]:=\{1,\dots,n\}$. An $[n,k,d]_q$ code is a linear code $\cC\subseteq\mathbb F_q^n$ of length $n$, dimension $k$, and minimum distance $d$. The Euclidean dual of $\cC$ is $\cC^\perp:=\{u\in\mathbb F_q^n:\langle u,c\rangle=0,\;\forall c\in\cC\}$, where $\langle\cdot,\cdot\rangle$ denotes the Euclidean inner product. We write $d(\cC)$ for the minimum distance of $\cC$ and $\wt(\cC)$ for the minimum Hamming weight of a nonzero codeword in $\cC$. For $I\subseteq[n]$, let $\pi_I:\mathbb F_q^n\to\mathbb F_q^{|I|}$ denote the projection onto the coordinates in $I$. The shortened code of $\cC$ on $I$ is $\sigma_I(\cC):=\{\pi_I(c):c\in\cC,\ \supp(c)\subseteq I\}$, where $\supp(c):=\{i\in[n]:c_i\neq0\}$. A $q$-ary quantum code of length $n$ is a subspace $\mathcal Q\subseteq(\mathbb C^q)^{\otimes n}$. If $\dim(\mathcal Q)=q^k$ and every erasure of at most $\delta-1$ qudits is correctable, then $\mathcal Q$ is called an $[[n,k,\delta]]_q$ quantum code. For a state $\rho$, let $\Tr_i(\rho)$ denote the partial trace over the $i$th qudit.\vspace{-0.15cm}

\subsection{Classical LRCs with intersecting recovery sets} \vspace{-0.05cm}

We begin by defining classical codes with locality.\vspace{-0.05cm}

\begin{definition}[cLRCs \cite{gopalan2012locality}]\label{def:cLRCs}
Let $\cC \subseteq \mathbb{F}_q^n$ be a linear code. For $i \in [n]$ and
$\beta \in \mathbb{F}_q$, define $\cC(i,\beta):=\{c\in\cC:\ c_i=\beta\}$. We say that $\cC$ has \emph{locality} $r$ if for every $i\in[n]$ there exists a
set $R_i\subseteq[n]$ with $i\in R_i$ and $|R_i|\le r+1$ such that,
\begin{align}
\hspace{-0.15cm}\pi_{R_i\setminus\{i\}}\big(\cC(i,\beta)\big)
\cap
\pi_{R_i\setminus\{i\}}\big(\cC(i,\beta')\big)
= \emptyset,\quad \forall \beta\neq\beta'.
\label{eq:LRCcond}
\end{align}
We refer to $R_i$ as a \emph{recovery set} for $i$.
\end{definition}

For linear codes, condition~\eqref{eq:LRCcond} is equivalent to the nonexistence of a codeword $c\in\cC$ such that $c_i\neq 0$ and $c_{R_i\setminus\{i\}}=0$, equivalently, $\sigma_{\{i\}}\big(\pi_{R_i}(\cC)\big)=\{0\}$. This characterization of classical locality is used in~\cite{galindo2026quantum}. Next, we recall cLRCs with multiple recovery sets and bounded pairwise
intersections.

\begin{definition}[$(r,t,x)$-cLRC \cite{kruglik2017one}]\label{def:cLRCs_tx}
A linear code $\cC \subseteq \mathbb F_q^n$ is called an $(r,t,x)$-cLRC if for every $i\in [n]$ there exist $t$ distinct subsets $\{R_i^{(j)}\}_{j=1}^t$ such that, for all $j\in[t]$, $i\in R_i^{(j)}$, $|R_i^{(j)}|\le r+1$, and the following conditions hold:
\begin{enumerate}
    \item For every pair of distinct indices $j,\ell\in[t]$,\vspace{-0.15cm}
    \begin{align}
        |R_i^{(j)} \cap R_i^{(\ell)}|\le x+1.
        \label{eq:rec_cond}
    \end{align}
    \item For every $j\in[t]$,\vspace{-0.15cm}
    \begin{align*}
    \sigma_{\{i\}}\big(\pi_{R_i^{(j)}}(\cC)\big)=\{0\}.
    \end{align*}
\end{enumerate}
We refer to $r$ as the \emph{locality}, $t$ as the \emph{availability}, by a slight abuse of
terminology, and $x$ as the \emph{intersection parameter}.\vspace{-0.15cm}
\end{definition}



\subsection{Quantum LRCs with intersecting recovery sets}


Quantum LRCs were introduced in~\cite{golowich2023quantum}, where \emph{disjoint} recovery sets ($x=0$) were shown to be impossible. Bu--Gu--Li~\cite{bu2025quantum} therefore introduced qLRCs with multiple \emph{intersecting} recovery sets, defined as follows.
\begin{definition}[$(r,t,x)$-qLRC \cite{bu2025quantum}]\label{def:qLRCs}
An $[[n,k,\delta]]_q$ quantum code $\mathcal Q$ is called an $(r,t,x)$-qLRC if, for every $i\in[n]$, there exist $t$ distinct subsets $\{R_i^{(j)}\}_{j=1}^t$ such that, for all $j\in[t]$, $i\in R_i^{(j)}$ with $|R_i^{(j)}|\le r+1$, and the following conditions hold:
\begin{enumerate}
    \item For every $j\in[t]$, there exists a recovery channel $\mathrm{Rec}_i^{(j)}$ acting on the qudits in $R_i^{(j)}\setminus\{i\}$ such that, for every state $\rho$ supported on $\mathcal Q$,
    \begin{align*}
    \bigl(\mathrm{Rec}_i^{(j)}\otimes
    \mathrm{id}_{[n]\setminus R_i^{(j)}}\bigr)\bigl(\Tr_i(\rho)\bigr)=\rho.
    \end{align*}
    \item For every pair of distinct indices $j,\ell\in[t]$,
    \begin{align*}
    |R_i^{(j)}\cap R_i^{(\ell)}|\le x+1.
    \end{align*}
\end{enumerate}
\end{definition}

In this paper, we restrict our attention to \rtxqs arising from the CSS construction~\cite{calderbank1998quantum}, which produces stabilizer codes~\cite{gottesman1997stabilizer} from suitable pairs of classical linear codes.

\begin{proposition}[CSS construction~\cite{calderbank1998quantum}]\label{prop:CSS}
Let $\cC_i$ be an $[n,k_i,d_i]_q$ code for $i=1,2$ satisfying $\cC_1^\perp\subseteq \cC_2$. Then, there exists a quantum code $\mathcal Q$ with parameters $[[n,\kappa,\delta]]_q$, where $\kappa=k_1+k_2-n$ and $\delta=\min\bigl\{\wt(\cC_2\setminus \cC_1^\perp),\,\wt(\cC_1\setminus \cC_2^\perp)\bigr\}$.\vspace{-0.08cm}
\end{proposition}

The following corollary is for the special case $\cC_1=\cC_2=\cC$. A linear code $\cC$ satisfying $\cC^\perp\subseteq\cC$ is called dual-containing.\vspace{-0.08cm}

\begin{corollary}[Dual-containing CSS construction~\cite{calderbank1998quantum}]\label{cor:css_selforth}
Let $\cC$ be an $[n,k,d]_q$ dual-containing linear code. Then the CSS construction produces a quantum code $[[n,\kappa,\delta]]_q$ with $\kappa=2k-n$ and $\delta=\wt(\cC\setminus\cC^\perp)\ge d$.
\end{corollary}
We write $\mathcal Q=\mathrm{CSS}(\cC_1,\cC_2)$ for the code in Proposition~\ref{prop:CSS}, and $\mathcal Q=\mathrm{CSS}(\cC)$ in the dual-containing case of Corollary~\ref{cor:css_selforth}. We call $\mathcal Q=\mathrm{CSS}(\cC)$ \emph{pure} if $\delta = \wt(\cC\setminus\cC^\perp)=d(\cC)$.\vspace{-0.15cm}

\section{\rtxqs from the CSS construction} \label{sec:css_constr}

In this section, we apply the CSS local recovery criterion of Galindo \emph{et al.}~\cite{galindo2026quantum} to the setting of $(r,t,x)$-qLRCs and identify the recovery set-condition needed to preserve the intersection parameter.
We first state the resulting criterion for a CSS code $\mathcal Q=\mathrm{CSS}(\cC_1,\cC_2)$ to be an \rtxq. \begin{lemma}\label{lem:css_constr} Let $\cC_1,\cC_2\subseteq \mathbb{F}_q^n$ be linear codes satisfying $\cC_1^\perp\subseteq \cC_2$, and let $\mathcal Q=\mathrm{CSS}(\cC_1,\cC_2)$. Then $\mathcal Q$ is an $(r,t,x)$-qLRC if and only if for every $i\in[n]$ there exist $t$ distinct subsets $\{R_i^{(j)}\}_{j=1}^t$ such that, for all $j, \ell\in[t]$, $i\in R_i^{(j)}$ and $|R_i^{(j)}|\le r+1$, 
\begin{align*}
|R_i^{(j)}\cap R_i^{(\ell)}|\le x+1 \qquad \text{for all } j\neq \ell,
\end{align*}
and, for all $j\in[t]$,\vspace{-0.15cm}
\begin{align*}
\sigma_{\{i\}}\!\big(\pi_{R_i^{(j)}}(\cC_1)\big)&=\sigma_{\{i\}}(\cC_2^\perp),\\
\sigma_{\{i\}}\!\big(\pi_{R_i^{(j)}}(\cC_2)\big)&=\sigma_{\{i\}}(\cC_1^\perp).
\end{align*}
\end{lemma}

\begin{proof}
By~\cite[Proposition~26]{galindo2026quantum}, a CSS code can recover the erasure of $I=\{i\}$ using the qudits in $J\setminus I$, where $I\subseteq J\subseteq [n]$, if and only if $\sigma_I\!\big(\pi_J(\cC_1)\big)=\sigma_I(\cC_2^\perp)$ and $\sigma_I\!\big(\pi_J(\cC_2)\big)=\sigma_I(\cC_1^\perp)$. Applying this criterion with $J=R_i^{(j)}$, for every $i\in[n]$ and $j\in[t]$, and including the size and intersection requirements in Definition~\ref{def:qLRCs}, gives the claim.\vspace{-0.15cm}
\end{proof}

We say that $\cC_1$ and $\cC_2$ have \emph{common recovery sets} if, for every $i\in[n]$, the same sets $\{R_i^{(j)}\}_{j=1}^t$ satisfy the recovery set conditions of Definition~\ref{def:cLRCs_tx} for both $\cC_1$ and $\cC_2$.  Under the additional assumption $d(\cC_i^\perp)\ge 2$ for $i=1,2$, the conditions in Lemma~\ref{lem:css_constr} reduce exactly to this common recovery set condition for the classical codes.\vspace{-0.15cm}
 
\begin{theorem}\label{thm:qrtx}
Let $\cC_\ell$ be an $[n,k_\ell,d_\ell]_q$ code with $d(\cC_\ell^\perp)\ge 2$ for $\ell=1,2$, and suppose that $\cC_1^\perp\subseteq \cC_2$. Then $\mathcal Q=\mathrm{CSS}(\cC_1,\cC_2)$ is an $(r,t,x)$-qLRC if and only if $\cC_1$ and $\cC_2$ are $(r,t,x)$-cLRCs with common recovery sets.\vspace{-0.1cm}
\end{theorem}

\begin{proof}
Since $d(\cC_\ell^\perp)\ge 2$, we have $\sigma_{\{i\}}(\cC_\ell^\perp)=\{0\}$ for every $i\in[n]$ and $\ell=1,2$. By Lemma~\ref{lem:css_constr}, the CSS recovery conditions therefore reduce to
\begin{align*}
\sigma_{\{i\}}\!\big(\pi_{R_i^{(j)}}(\cC_\ell)\big)=\{0\},
\qquad
\ell=1,2,\;\; j\in[t].
\end{align*}
Thus the sets $\{R_i^{(j)}\}_{j=1}^t$ appearing in Lemma~\ref{lem:css_constr} satisfy the classical recovery-set condition for both $\cC_1$ and $\cC_2$, i.e., they have common recovery sets. Conversely, suppose that $\cC_1$ and $\cC_2$ are \rtxcs with common recovery sets. Then these same sets satisfy the reduced CSS recovery conditions above. Together with the size and intersection requirements, Lemma~\ref{lem:css_constr} implies that $\mathcal Q$ is an $(r,t,x)$-qLRC. \vspace{-0.15cm}
\end{proof}

The common recovery set requirement in Theorem~\ref{thm:qrtx} is needed to preserve the intersection parameter $x$. If $\cC_1$ and $\cC_2$ use different recovery families, one may form CSS recovery sets by taking unions of the corresponding classical recovery sets. However, intersections between such unions also involve cross-intersections between recovery sets from $\cC_1$ and $\cC_2$, which are not controlled by the original $(r,t,x)$ conditions. \vspace{-0.25cm}

\section{Construction of CSS \rtxqs} \label{sec:wzl_gen}

In this section, we construct explicit binary dual-containing $(r,t,x)$-cLRCs, which then yield \rtxqs through the CSS construction. We use the subset-inclusion code family considered in~\cite{wilson1990diagonal,marin2026binary}, which contains the WZL construction~\cite{wang2015achieving} for $(r,t,0)$-cLRCs as a special case; hence we refer to it as the \emph{generalized WZL construction}. We show that this family yields $(r,t,x)$-cLRCs with explicit locality, availability, and intersection parameter. Since dual-containment is not automatic, we then identify the dual-containing subfamilies before applying the CSS construction.\vspace{-0.15cm}

\subsection{Dual-containing generalized WZL construction}

We now present the subset-inclusion code family \cite{marin2026binary} that we refer to as the \emph{generalized WZL construction}. Let $m,s,\alpha$ be integers such that $\alpha<s<m$. We define a binary matrix $H_{m,s,\alpha}$ over $\mathbb F_2$ with $\binom{m}{s-\alpha}$ rows and $\binom{m}{s}$ columns. The rows are indexed by the $(s-\alpha)$-subsets of $[m]$, and the columns are indexed by the $s$-subsets of $[m]$. For an $(s-\alpha)$-subset $E\subseteq[m]$ and an $s$-subset $F\subseteq[m]$, the entry in row $E$ and column $F$ is\vspace{-0.15cm}
\begin{align*}
h_{E,F}
=
\begin{dcases}
1 & \text{if } E\subseteq F,\\
0 & \text{otherwise}.
\end{dcases}
\end{align*}
Let $\cC_{m,s,\alpha}$ denote the binary linear code with parity-check matrix $H_{m,s,\alpha}$. This family of subset-inclusion codes was considered in~\cite{wilson1990diagonal,marin2026binary}. Apart from the special case $\alpha=1$, which recovers the WZL construction of binary $(r,t,0)$-cLRCs~\cite{wang2015achieving}, these codes have not, to the best of our knowledge, been studied from the viewpoint of \rtxcs. We show below that $\cC_{m,s,\alpha}$ is an \rtxc with explicitly computable locality, availability, and intersection parameter.

\begin{theorem}\label{thm:parameters}
The binary linear code $\cC_{m,s,\alpha}$ with parameters $[n,k,d]_2$ is an $(r,t,x)$-cLRC with $n = \binom{m}{s}$, $t = \binom{s}{\alpha}$, and  
\begin{align*}
r = \binom{m-s+\alpha}{\alpha}-1,\quad \;x = \binom{m-s+\alpha-1}{\alpha-1}-1.
\end{align*}
Moreover, if $m \ge 2s-\alpha+1$, then
\begin{align*}
s-\alpha+2 \le d \le 2^{\,s-\alpha+1},
\end{align*}
and if $m \ge 2s-\alpha$, then\vspace{-0.15cm}
\begin{align*}
k
=
\binom{m}{s}
-
\sum_{\substack{0\le i\le s-\alpha\\ \binom{s-i}{\alpha}\not\equiv 0 \!\!\!\pmod 2}}
\left(\binom{m}{i}-\binom{m}{i-1}\right).
\end{align*}
\end{theorem}

\begin{proof}
We first determine the locality. Fix a coordinate indexed by an $s$-subset $F\subseteq[m]$. Any row indexed by an $(s-\alpha)$-subset $E\subseteq F$ gives a parity-check whose support is a recovery set for the coordinate $F$. The support of this row is the set of all $s$-subsets of $[m]$ containing $E$, so it has size $\binom{m-s+\alpha}{\alpha}$. 
Thus, $r=\binom{m-s+\alpha}{\alpha}-1$. The availability is the number of rows containing the coordinate $F$, namely the number of $(s-\alpha)$-subsets of $F$. Hence $t=\binom{s}{s-\alpha}=\binom{s}{\alpha}$. Next, we compute the intersection parameter. Consider two distinct rows indexed by $(s-\alpha)$-subsets $E,E'\subseteq[m]$. A column indexed by an $s$-subset $F$ contains $1$ in both rows if and only if $E\cup E'\subseteq F$. Hence the number of common $1$-positions is the number of $s$-subsets containing $E\cup E'$, namely
\begin{align*}
\binom{m-|E\cup E'|}{s-|E\cup E'|}.
\end{align*}
This quantity is maximized when $|E\cup E'|$ is minimized. Since $E$ and $E'$ are distinct $(s-\alpha)$-subsets, the minimum is $s-\alpha+1$, attained when they differ in exactly one element. Substituting this value gives $x+1=\binom{m-s+\alpha-1}{\alpha-1}$. Finally, the bounds on the minimum distance follow from~\cite[Corollary~1]{marin2026binary}, and the dimension formula follows from~\cite[Theorem~1]{wilson1990diagonal}.
\end{proof}

The following theorem characterizes exactly when $\cC_{m,s,\alpha}$ is dual-containing.
\begin{theorem}\label{thm:dual_cont}
Let $m,s,\alpha$ be integers with $\alpha<s<m$. Then $\cC_{m,s,\alpha}$ is dual-containing iff 
\begin{align*}
\binom{m-u}{s-u}\equiv 0 \pmod 2
\end{align*}
for every $u\in\{s-\alpha,s-\alpha+1,\dots,\min\{2(s-\alpha),m\}\}$.
\end{theorem}

\begin{proof}
The code $\cC_{m,s,\alpha}$ has parity-check matrix $H_{m,s,\alpha}$, so $\cC_{m,s,\alpha}$ is dual-containing iff $H_{m,s,\alpha}H_{m,s,\alpha}^T=0 \pmod 2$. Let $h_E$ and $h_{E'}$ be two rows indexed by $(s-\alpha)$-subsets $E,E'\subseteq[m]$. A column indexed by an $s$-subset $S\subseteq[m]$ contributes $1$ to $h_E h_{E'}^T$ iff $E\subseteq S$ and $E'\subseteq S$, equivalently $E\cup E'\subseteq S$. Hence $h_E h_{E'}^T$ equals, modulo $2$, the number of $s$-subsets of $[m]$ containing $E\cup E'$. Writing $u:=|E\cup E'|$, we obtain
\begin{align*}
h_E h_{E'}^T \equiv \binom{m-u}{s-u}\pmod 2.
\end{align*}
Since $|E|=|E'|=s-\alpha$, the possible union sizes satisfy $s-\alpha\le u\le \min\{2(s-\alpha),m\}$. Conversely, every integer in this range occurs as $|E\cup E'|$ for some pair of $(s-\alpha)$-subsets $E,E'\subseteq[m]$, by choosing $|E\cap E'|=2(s-\alpha)-u$. Therefore, $H_{m,s,\alpha}H_{m,s,\alpha}^T=0 \pmod 2$ iff the stated congruence holds for all such $u$.
\end{proof}
\begin{remark}
If $s\ge 2\alpha$, then $u=s$ belongs to the range in Theorem~\ref{thm:dual_cont}, and the corresponding condition becomes $\binom{m-s}{0}\equiv 0 \pmod 2$, which is impossible. Hence $s<2\alpha$ is necessary for dual-containment.
\end{remark}

A related construction of binary $(r,t,x)$-cLRCs from subspace-inclusion matrices was given in~\cite{teng2021constructions}. However, the parity-check matrices have odd row weights and hence are not self-orthogonal, so they cannot be used in the dual-containing CSS
construction.

In the next subsection, we apply Theorem~\ref{thm:dual_cont} to the generalized WZL family to obtain CSS $(r,t,x)$-qLRCs.

\subsection{CSS codes from the generalized WZL construction}

We now apply the dual-containing generalized WZL construction within the CSS framework. The following corollary gives the resulting \rtxq parameters.

\begin{corollary}\label{cor:wzl_css}
Let $m,s,\alpha$ be integers with $\alpha<s<m$ satisfying the conditions of Theorem~\ref{thm:dual_cont}. Then $\mathcal Q_{m,s,\alpha}=\mathrm{CSS}(\cC_{m,s,\alpha})$ is an $(r,t,x)$-qLRC, where $r,t,x$ are as in Theorem~\ref{thm:parameters}. Moreover, if $m\ge 2s-\alpha$, then 
\begin{align*}
\kappa
=
\binom{m}{s}
-
2\sum_{\substack{0\le i\le s-\alpha\\
\binom{s-i}{\alpha}\not\equiv0\!\!\!\pmod2}}
\left(\binom{m}{i}-\binom{m}{i-1}\right),
\end{align*}
and if $m\ge 2s-\alpha+1$, then $\delta \ge d(\cC_{m,s,\alpha}) \ge s-\alpha+2$.
\end{corollary}

\begin{proof}
By Theorem~\ref{thm:dual_cont}, $\cC_{m,s,\alpha}$ is dual-containing. Since each column of $H_{m,s,\alpha}$ has weight $\binom{s}{\alpha}>0$, $\cC_{m,s,\alpha}$ has no weight-one codeword, so $d(\cC_{m,s,\alpha})\ge 2$. As $\cC_{m,s,\alpha}^\perp\subseteq\cC_{m,s,\alpha}$, it follows that $d(\cC_{m,s,\alpha}^\perp)\ge 2$. By Theorem~\ref{thm:parameters}, $\cC_{m,s,\alpha}$ is an $(r,t,x)$-cLRC with the stated parameters, so Theorem~\ref{thm:qrtx} implies that $\mathcal Q_{m,s,\alpha}$ is an $(r,t,x)$-qLRC with the same recovery sets.
\end{proof}

Table~\ref{tab:best} lists representative choices of $(m,s,\alpha)$ satisfying Theorem~\ref{thm:dual_cont}, together with the resulting classical and CSS quantum parameters. The examples show that the generalized WZL construction gives binary \rtxqs with high rates and nontrivial guaranteed distances. The listed distances, or lower bounds, follow from~\cite[Propositions~3 and~4, Theorems~4--6]{marin2026binary}, in particular, the cases $(10,3,2)$, $(12,6,4)$, and $(14,7,4)$ attain the upper bound $d\le 2^{s-\alpha+1}$.

\begin{table}[t]
\centering
\caption{Parameters from the generalized WZL CSS construction.}\setlength{\tabcolsep}{3.2pt}
\scriptsize
\label{tab:best}
\renewcommand{\arraystretch}{1.2}
\begin{tabular}{|c|c|c|c|c|}
\hline
$\boldsymbol{(m,s,\alpha)}$ & $\boldsymbol{(r,t,x)}$ & $\boldsymbol{[n,k,d]}$ & $\boldsymbol{[[n,\kappa,\delta]]}$ & $\boldsymbol{R}$ \\
\hline
$(7,4,3)$     & $(19,4,9)$      & $[35,29,3]$        & $[[35,23,\ge 3]]$      & 0.66 \\
$(10,3,2)$    & $(35,3,7)$      & $[120,110,4]$      & $[[120,100,\ge 4]]$    & 0.83 \\
$(12,5,3)$    & $(119,10,35)$   & $[792,738,6]$      & $[[792,684,\ge 6]]$    & 0.86 \\
$(12,8,6)$    & $(209,28,125)$  & $[495,430,\ge 7]$  & $[[495,365,\ge 7]]$    & 0.74 \\
$(12,6,4)$    & $(209,15,83)$   & $[924,858,8]$      & $[[924,792,\ge 8]]$    & 0.86 \\
$(14,8,5)$    & $(461,56,209)$  & $[3003,2717,9]$    & $[[3003,2431,\ge 9]]$  & 0.81 \\
 $(14,10,7)$   & $(329,120,209)$ & $[1001,728,11]$    & $[[1001,455,\ge 11]]$  & 0.45 \\
$(14,9,6)$    & $(461,84,251)$  & $[2002,1652,\ge 14]$ & $[[2002,1302,\ge 14]]$ & 0.65 \\
$(14,7,4)$    & $(329,35,119)$  & $[3432,3068,16]$   & $[[3432,2704,\ge 16]]$ & 0.79 \\
\hline
\end{tabular} \vspace{-0.25cm}
\end{table}

\section{Bounds for CSS $(r,t,x)$-qLRCs} \label{sec:upper_bounds}

In this section, we derive dimension and rate bounds for dual-containing CSS \rtxqs from known bounds for classical \rtxcs. We also obtain distance bounds for pure CSS codes. Finally, we derive a Singleton-like dimension bound for exact dual-containing CSS \rtxqs that depends on the underlying classical minimum distance.

\subsection{General CSS dimension, rate, and distance bounds}

Following the approach of~\cite{galindo2026quantum}, we derive bounds for dual-containing CSS \rtxqs from bounds for the underlying classical \rtxcs. Let $d^{(r,t,x)}_{q,\mathrm{opt}}(n,k)$ and $k^{(r,t,x)}_{q,\mathrm{opt}}(n,d)$ denote, respectively, the maximum possible minimum distance and maximum possible dimension among $q$-ary linear \rtxcs with the specified parameters.

 \noindent\smash{\rule{0pt}{0.4pt}}\par\vspace*{-\baselineskip}
 \vspace*{0.4pt}
\begin{theorem}\label{thm:css_bounds}
Let $\cC$ be an $[n,k,d]_q$ code such that $d(\cC^\perp)\ge 2$, $\cC^\perp\subseteq \cC$, and $\cC$ is an $(r,t,x)$-cLRC. Let $\mathcal Q=\mathrm{CSS}(\cC)$ have parameters $[[n,\kappa,\delta]]_q$, with $\delta\ge d$. Then
\begin{align*}
\kappa &\le 2k^{(r,t,x)}_{q,\mathrm{opt}}(n,d)-n, \\
d &\le d^{(r,t,x)}_{q,\mathrm{opt}}\Bigl(n,\frac{n+\kappa}{2}\Bigr).
\end{align*}
\end{theorem}

\begin{proof}
By Corollary~\ref{cor:css_selforth}, $\kappa=2k-n$. Since $\cC$ is a $q$-ary linear $(r,t,x)$-cLRC of length $n$, dimension $k$, and minimum distance $d$, the definition of $k^{(r,t,x)}_{q,\mathrm{opt}}(n,d)$ gives $k\le k^{(r,t,x)}_{q,\mathrm{opt}}(n,d)$. Hence $\kappa\le 2k^{(r,t,x)}_{q,\mathrm{opt}}(n,d)-n$. Similarly, since $k=(n+\kappa)/2$, the definition of $d^{(r,t,x)}_{q,\mathrm{opt}}(n,k)$ gives the second bound.
\end{proof}

\begin{remark}
Puncturing-based arguments for single-recovery-set qLRCs~\cite{li2025improved,luo2025bounds,li2025optimal} do not directly extend to \rtxqs. Although locality and bounded intersections may be preserved, availability need not be: distinct recovery sets can coincide after puncturing.
\end{remark}


Next, we recall two known bounds for \rtxcs. Every $(r,t,x)$-cLRC with parameters $[n,k,d]_q$ satisfies the alphabet-dependent distance bound~\cite{kruglik2018distance}, 
\begin{align*}
d \le \min_{1\le i \le k-r} \frac{q^i-q^{i-1}}{q^i - 1} \Big( n-(k-i)-\Big\lfloor  \frac{k-1-i}{r-1}\Big\rfloor\Big),
\end{align*}
and the dimension bound~\cite{kruglik2017one}, $k \le n(1-p(r,t,x))$, where
$s_L=\min\{L,\lfloor r/x\rfloor+1\}$ and
\begin{align*}
p(r,t,x) &= \sum_{\substack{1\le L\le t\\ L\text{ odd}}} \frac{\binom{t}{L}}{Lr+1}- \sum_{\substack{1\le L\le t\\ L\text{ even}}} \frac{2\binom{t}{L}}{2+s_L(2r-(s_L-1)x)}.
\end{align*}
This yields the following explicit CSS bounds.\vspace{-0.1cm}
\begin{corollary}\label{cor:css_boundss}
Let $\cC$ be as in Theorem~\ref{thm:css_bounds}, and let $\mathcal Q=\mathrm{CSS}(\cC)$ have parameters $[[n,\kappa,\delta]]_q$ with $\delta \ge d$. Then
\begin{align*}
\kappa &\le n(1-2p(r,t,x)),\\
d &\le \min_{1\le i \le \frac{n+\kappa}{2}-r} \frac{q^i-q^{i-1}}{2(q^i - 1)} \Big( n-(\kappa-2i) - 2\Big\lfloor \frac{\frac{n+\kappa}{2}-1-i}{r-1} \Big\rfloor \Big).
\end{align*}
\end{corollary}

\begin{proof}
Apply Theorem~\ref{thm:css_bounds} to the classical bounds.\vspace{-0.1cm}
\end{proof}
The first bound in Corollary~\ref{cor:css_boundss} gives the rate bound $R:= \frac{\kappa}{n}\le 1-2p(r,t,x)$.

\begin{remark}
The distance bounds above are stated for the underlying classical distance $d$. If $\mathcal Q=\mathrm{CSS}(\cC)$ is pure, then $\delta=d$, and the same bounds apply to the quantum distance.
\end{remark}



\subsection{Dimension bound for exact CSS codes}

We now specialize to exact $(r,t,x)$-cLRCs. We first derive a classical analogue of the Singleton-like dimension bound in~\cite[Theorem~15]{bu2025quantum}, and then translate it to dual-containing CSS codes. Recall the definition of exact $(r,t,x)$-cLRCs/qLRCs.

\begin{definition}[\emph{Exact} $(r,t,x)$-qLRC/cLRC \cite{bu2025quantum}]
\label{def:exactc}
An $(r,t,x)$-qLRC/cLRC is called \emph{exact} if its recovery sets additionally satisfy $|R_i^{(j)}|=r+1$ for all $i\in[n]$ and $j\in[t]$, $|R_i^{(j)}\cap R_i^{(\ell)}|=x+1$ for all $i\in[n]$ and all distinct $j,\ell\in[t]$, and for all $i\in[n]$
\begin{align*}
R_i^{(j)}\cap R_i^{(\ell)}\cap R_i^{(f)}=\{i\}, \quad  \forall\:j,\ell,f\in[t]\ \text{pairwise distinct}.
\end{align*}
\end{definition}

The fixed recovery-set sizes and intersection pattern lead to a Singleton-like
dimension bound depending on the minimum distance. The following bound is
obtained by adapting the proof of~\cite[Theorem~15]{bu2025quantum} to exact
$(r,t,x)$-cLRCs.

\begin{proposition}\label{prop:dim_exact}
Let $\cC$ be a linear $[n,k,d]_q$ exact $(r,t,x)$-cLRC. Then \vspace{-0.15cm}
\begin{align*}
k\le n-(d-1)-\bar{N}\!\left(n,r,d,\left\lceil n\,p_e(r,t,x)\right\rceil\right),
\end{align*}
where 
\begin{align*} p_e(r,t,&x) =\sum_{L=1}^{t}(-1)^{L+1}\binom{t}{L} \frac{2}{2+L\bigl(2r-(L-1)x\bigr)}, \\ \bar{N}(n,r,d,M&) :=\max \Big \{ N\in\{0,\dots,M\} : N(r+1) \\ &-\frac{N(N-1)}{M(M-1)} \big (M(r+1)-n\big ) \le n-(d-1) \Big \}. \end{align*}
\end{proposition}
\begin{proof}[Proof sketch]
Let $\{R_i^{(j)}\}_{i\in[n],\,j\in[t]}$ be the exact recovery sets of $\cC$.
By~\cite[Lemma~14]{bu2025quantum}, there exists a set
$U:=\{i_1,\dots,i_T\}\subseteq[n]$, where
$T=(d-1)+\bar{N}\!\left(n,r,d,\left\lceil n\,p_e(r,t,x)\right\rceil\right)$,
such that for every $d\le \ell\le T$, one can choose
$a_\ell\in[t]$ satisfying
$R_{i_\ell}^{(a_\ell)}\cap \{i_1,\dots,i_{\ell-1}\}=\emptyset$.
We recover locally the coordinates $i_T,i_{T-1},\dots,i_d$ in this reverse
order, since at each step the chosen recovery set is disjoint from the erased
coordinates that remain. The coordinates $\{i_1,\dots,i_{d-1}\}$ are then
recoverable globally because $\cC$ has minimum distance $d$. Thus the symbols in
$U$ are determined by the symbols in $[n]\setminus U$, so puncturing on
$[n]\setminus U$ is injective on $\cC$ and therefore $k\le n-|U|$.
\end{proof}


This gives the following Singleton-like dimension bound.

\begin{corollary}\label{cor:CSS_exact}
Let $\cC$ be as in Theorem~\ref{thm:css_bounds}. If $\cC$ is an exact $(r,t,x)$-cLRC, then $\mathcal Q=\mathrm{CSS}(\cC)$ is an exact $(r,t,x)$-qLRC with parameters $[[n,\kappa,\delta]]_q$, where $\delta\ge d$, and
\begin{align*}
\kappa \le n-2(d-1) -2\bar{N}\!\left(n,r,d,\left\lceil n\,p_e(r,t,x)\right\rceil\right).
\end{align*}
\end{corollary}

\begin{proof}
By Theorem~\ref{thm:qrtx}, $\mathcal Q$ has the same recovery sets as $\cC$, and is therefore exact. The bound follows by applying Proposition~\ref{prop:dim_exact} with $k=(n+\kappa)/2$ and rearranging. 
\end{proof}


\section{Comparison of exact CSS \rtxqs} \label{sec:comp_exact}

We now examine an exact subfamily of the generalized WZL CSS construction. This subfamily allows a direct comparison with the exact CSS bound from Corollary~\ref{cor:CSS_exact}, the exact quantum bound of Bu--Gu--Li~\cite{bu2025quantum}, and their explicit exact construction. Note that this subfamily is not representative of the full generalized WZL CSS parameter range.

\begin{example}\label{ex:a2s3_family}
Let $(s,\alpha)=(3,2)$ and let $m=4\ell+2$ for some $\ell\ge 1$. Then $\cC_{m,3,2}$ is a dual-containing exact $(r,t,x)$-cLRC with locality $r=\binom{m-1}{2}-1$, availability $t=3$, intersection parameter $x=m-3$, and parameters $\big[\binom{m}{3},\,\binom{m}{3}-m,\,4\big]$. Hence $\mathcal Q_{m,3,2}=\mathrm{CSS}(\cC_{m,3,2})$ is an exact pure $(r,t,x)$-qLRC with parameters $\big[\big[\binom{m}{3},\,\binom{m}{3}-2m,\,4\big]\big]$.
\end{example}

For a coordinate indexed by $F=\{u,v,w\}$, the recovery sets induced by $\{u\}$, $\{v\}$, and $\{w\}$ have size $\binom{m-1}{2}$, any two intersect in $m-2$ coordinates, and all three intersect only in the coordinate indexed by $F$. Hence the family is exact.

It remains to justify purity. Since $\cC_{m,3,2}^{\perp}$ is the row span of $H_{m,3,2}$, every nonzero word is the sum of the rows indexed by a nonempty set $A\subseteq[m]$. If $|A|=a$, then a triple $F\subseteq[m]$ contributes precisely when $|A\cap F|$ is odd, so the resulting word has weight $\binom{a}{1}\binom{m-a}{2}+\binom{a}{3}$. For $m\ge 6$, this weight is larger than $4$ for every $1\le a\le m$, so $d(\cC_{m,3,2}^{\perp})>4$. Hence the CSS code is pure and has distance $\delta=4$.

Table~\ref{tab:a2s3_family_params} compares the exact $(3,2)$-family with the bounds derived above. The achieved dimensions are close to both the general exact bound $\kappa_{\mathrm e}^*$ and the CSS-specific bound $\tilde{\kappa}_{\mathrm e}^*$, which coincide from $m=14$ onward. The achieved rate approaches the rate bound $R^*$ as $m$ grows. Since the family is pure, the distance bound from Corollary~\ref{cor:css_boundss} applies to its quantum distance. However, the family has fixed distance $\delta=4$, while the upper bound $\delta^*$ grows with $m$.

\begin{table}[t]
\centering
\setlength{\tabcolsep}{1.5pt}
\caption{Parameters of the exact $(3,2)$-family. Here $\kappa_{\mathrm e}^*$ is the general exact bound~\cite[Thm.~15]{bu2025quantum}, $\tilde{\kappa}_{\mathrm e}^*$ is the CSS exact bound (Cor.~\ref{cor:CSS_exact}), and $R^*$ and $\delta^*$ are the rate and distance bounds from Cor.~\ref{cor:css_boundss}.}
\label{tab:a2s3_family_params}
\scriptsize
\renewcommand{\arraystretch}{1.15}
\begin{tabular}{|c|c|c|c|c|c|c|c|c|}
\hline
$\boldsymbol{m}$ & $\boldsymbol{(r,t,x)}$ & $\boldsymbol{[n,k,d]}$ & $\boldsymbol{[[n,\kappa,\delta]]}$ & $\boldsymbol{R}$ & $\boldsymbol{\kappa^*_{\mathrm{e}}}$ & $\boldsymbol{\tilde{\kappa}_{\mathrm{e}}^{*}}$ & $\boldsymbol{R^*}$ & $\boldsymbol{\delta^*}$ \\
\hline
$6$   & $(9,3,3)$      & $[20,14,4]$      & $[[20,8, 4]]$       & 0.400 & 11   & 10   & 0.704 & $4$ \\
$10$  & $(35,3,7)$     & $[120,110,4]$    & $[[120,100, 4]]$    & 0.833 & 105  & 104  & 0.908 & $5$ \\
$14$  & $(77,3,11)$    & $[364,350,4]$    & $[[364,336, 4]]$    & 0.923 & 344  & 344  & 0.956 & $7$ \\
$18$  & $(135,3,15)$   & $[816,798,4]$    & $[[816,780, 4]]$    & 0.956 & 790  & 790  & 0.974 & $9$ \\
$22$  & $(209,3,19)$   & $[1540,1518,4]$  & $[[1540,1496, 4]]$  & 0.972 & 1508 & 1508 & 0.983 & $10$ \\
$26$  & $(299,3,23)$   & $[2600,2574,4]$  & $[[2600,2548, 4]]$  & 0.980 & 2564 & 2564 & 0.988 & $11$ \\
\hline
\end{tabular}\vspace{-0.35cm}
\end{table}
We next compare the exact subfamily in Example~\ref{ex:a2s3_family} with the
explicit exact construction of Bu--Gu--Li~\cite{bu2025quantum}. Their
construction gives exact $(r,t,1)$-qLRCs parametrized by an integer $\mu$,
with locality $r=2^{\mu+1}-1$, availability $t=2^{\mu+1}$, length
$n=7^{\,2^{\mu-1}}2^{\,2^{\mu-1}-1}$, dimension at least
$2^{\,2^{\mu-1}-1}$, and distance at least $3$. It achieves the smallest
possible intersection parameter $x=1$ and large availability, but the resulting
rate lower bound is $7^{-2^{\mu-1}}$, which vanishes as $\mu$ grows. In
contrast, the exact subfamily in Example~\ref{ex:a2s3_family} has fixed
availability $t=3$, larger recovery set intersections, and distance
$\delta=4$, while its rate is
$R_m = 1-\frac{12}{(m-1)(m-2)}$, which tends to $1$. The two constructions
therefore exhibit different tradeoffs: smaller intersections and larger
availability in~\cite{bu2025quantum} versus higher rates and certified
minimum distance $4$ in the exact subfamily above.


\section{Conclusion}

We studied CSS $(r,t,x)$-qLRCs and showed that, under a dual
minimum-distance assumption, the $(r,t,x)$-qLRC property is equivalent to the
underlying classical codes being \rtxcs with common recovery sets. Using the
subset-inclusion construction, we then constructed explicit infinite families
of binary CSS $(r,t,x)$-qLRCs. We also derived dimension and rate bounds for
CSS $(r,t,x)$-qLRCs by translating known bounds for classical \rtxcs, obtained
minimum-distance bounds in the pure case, and proved a Singleton-like dimension
bound in the exact case. Future work includes extending these constructions
beyond the binary CSS setting and identifying parameter regimes in which
inclusion-matrix CSS codes yield useful quantum LDPC families.
\bibliographystyle{ieeetr}
\bibliography{ref}

\end{document}